\documentclass[letterpaper, 10 pt, conference]{ieeeconf}                                                          % if you need a4paper
\IEEEoverridecommandlockouts                              % This command is only
\usepackage{graphics} % for pdf, bitmapped graphics files
\usepackage{epsfig} % for postscript graphics files
\usepackage{times} % assumes new font selection scheme installed
\usepackage{cite}

\usepackage{amsmath} % assumes amsmath package installed
\usepackage{amssymb}  % assumes amsmath package installed

\usepackage{color,soul}

\graphicspath{ {images/} }

\usepackage{algorithm}
\usepackage{algpseudocode}

\usepackage{amsthm}                               
\theoremstyle{remark}

\newtheorem{theorem}{Theorem}
\newtheorem{lemma}{Lemma}

\makeatletter
\newcommand{\removelatexerror}{\let\@latex@error\@gobble}
\makeatother

\title{\LARGE \bf
A Novel Technique for Rejecting Non-Aircraft Artefacts in Above Horizon Vision-Based Aircraft Detection}

\author{Jasmin James$^1$,  Jason J. Ford$^1$ and Timothy L. Molloy$^2$% <-this % stops a space
\thanks{$^1$J. James and J. J. Ford are with the School of Electrical Engineering and Computer Science, Queensland University of Technology, 2 George St, Brisbane QLD, 4000 Australia. {\tt\small jasmin.martin@qut.edu.au, j2.ford@qut.edu.au}
\newline $^2$ T. Molloy is with the Department of Electrical and Electronic Engineering, University of Melbourne, Parkville VIC, 3010 Australia {\tt\small tim.molloy@unimelb.edu.au}.}%
}

\begin{document}
\maketitle
\thispagestyle{empty}
\pagestyle{empty}

\begin{abstract}
Unmanned aerial vehicle (UAV) operations are steadily expanding into many important applications. A key technology for better enabling their commercial use is an onboard  sense and avoid (SAA) technology which can detect potential mid-air collision threats in the same manner expected from a human pilot. Ideally, aircraft should be detected as early as possible whilst maintaining a low false alarm rate, however, textured clouds and other unstructured terrain make this trade-off a challenge. In this paper we present a new technique for the modelling and detection of aircraft above the horizon that is able to penalise non-aircraft artefacts (such as textured clouds and other unstructured terrain). 
%We cast the vision-based aircraft detection problem as a quickest intermittent signal detection and identification (ISDI) problem where we are able to introduce a penalty  on misidentification (e.g. a penalty on detecting non-aircraft artefacts such as cloud features). 
We evaluate the performance of our proposed system on flight data of a Cessna 172 on a near collision course encounter with a ScanEagle UAV data collection aircraft. By penalising non-aircraft artefacts we are able to demonstrate, for a zero false alarm rate, a mean detection range of 2445m corresponding to an improvement in  detection ranges by  9.8\% (218m). 
\end{abstract}

%%%%%%%%%%%%%%%%%%%%%%%%%%%%%%%%%%%%%%%%%%%%%%%%%%%%%%%%%%%%%%%%%%%%%%%%%%%%%%%%%%%%%%%%%%%%%%%%%%%%%%
\section{Introduction}
The unmanned aerial vehicle (UAV) market worldwide is projected to grow by US \$47.8 Billion by 2025, with a compounded growth of 18.8\%. This growth is driven by the increasing use of UAVs in various commercial applications, such as monitoring, surveying and mapping, precision agriculture, aerial remote sensing, product delivery and many more \cite{Researchandmarkets,PwC,Dalamagkidis2008}. The consistent increase in the global market has propelled efforts to ensure that routine, standard and flexible UAV operations are integrated into the national airspace such that they do not compromise the existing safety levels  \cite{Clothier2015}. The risk of mid-air collision is an important safety concern that is both faced and posed by UAVs.
The  capability to avoid mid-air collisions would allow UAVs to more routinely operate in common airspace \cite{Zarandy2015,Morris20005}.

In order to reduce the risk of a mid-air collision, the national airspace is strictly regulated with several safety layers \cite{Clothier2015,Nussberger2014, Zarandy2015}. The first few layers involve operational procedures, air traffic management and cooperative collision avoidance systems. The final layer is for potential mid-air collision threats that are not caught by the other layers e.g. aircraft that are not communicating their presence. In these situations non-cooperative collision avoidance is necessary.  For a manned aircraft this final safety layer involves a pilot visually seeing and then avoiding a collision threat. For UAVs this final safety layer involves sensing and avoiding potential non-cooperative, mid-air collision threats with a proficiency matching or exceeding that of human pilots. 
A general guideline for how well human pilots detect potential collision threats is given in \cite{andrews1989modeling} where pilots, who were alerted to the presence of potential collision threat, were able to detect them with an $86\%$ success rate at a median range of $2593$m.

For sense and avoid (SAA), in small to medium sized UAVs, machine vision has been established as a potential technology as vision sensors have power, cost, size and weight benefits over other sensing approaches \cite{Mcfadyen2016}.  To meet the approximate human guidelines (as well as providing sufficient time for avoidance manoeuvres)  it is desirable to  detect aircraft as early as possible in an image sequence. At these ranges aircraft appear in vision sensors as a very small number (approx 1-10) of  locally dim pixels that poorly contrast with the background (see Figure \ref{fig:exampleAircraft} for an example of aircraft size). In this paper we will focus on the ``sense'' aspect of SAA, specifically vision-based aircraft detection. There have been a variety of collision avoidance strategies  to address  the ``avoid'' aspect suitable for use with vision-based aircraft detection approaches (see \cite{Gunasinghe,Mcfadyen2016} and references therein), however that is out of scope for this paper.
 
 \begin{figure}
\begin{center}
\includegraphics[scale=0.55,trim={0.0cm 0cm 0.0cm 0.0cm}, clip]{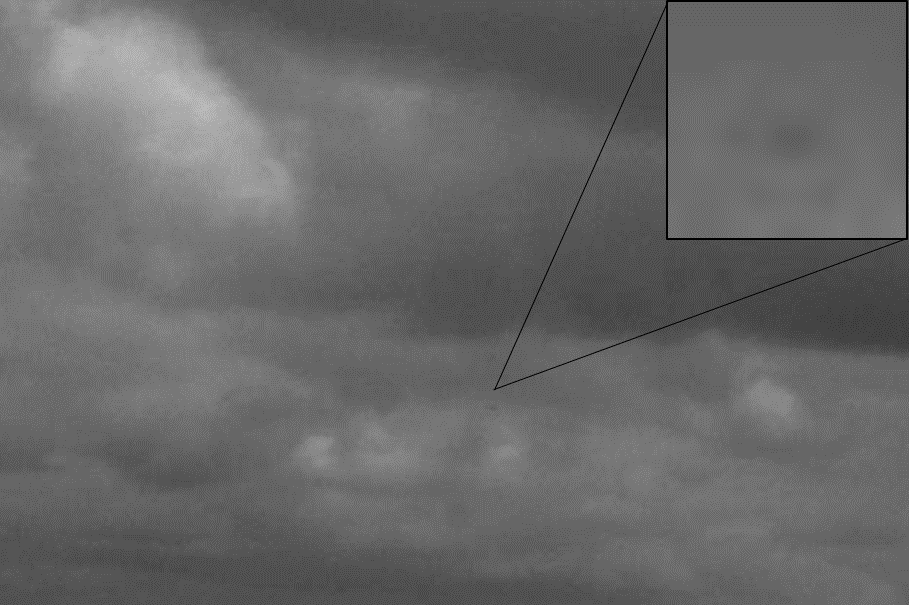}
  \caption{An example of the above horizon conditions and aircraft target size which we are trying to detect. Note the texture in the clouds and the similarities between the aircraft and cloud features}
\label{fig:exampleAircraft}
\end{center}
\end{figure}
 
 %%for some reason this has to be up here
  \begin{figure*}
\begin{center}
\includegraphics[scale=0.55,trim={0.0cm 0cm 0.0cm 0.0cm}, clip]{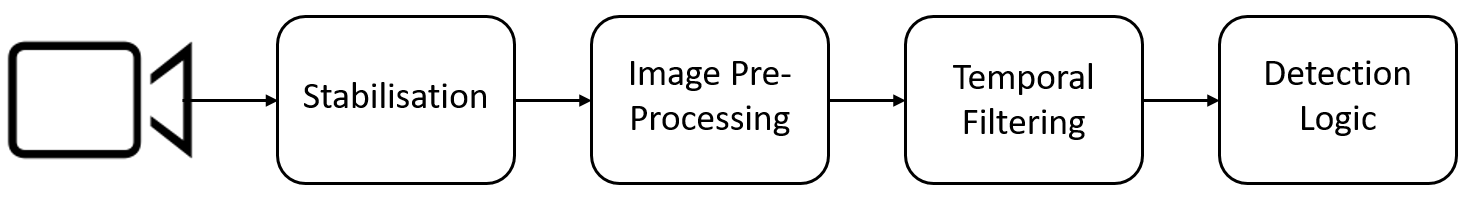}
  \caption{Overview of our multi-stage detection system. An image sensor captures an image which is then stabilised. Morphology is used for the image pre-processing stage and hidden Markov model (HMM) filtering for the temporal filtering. For the detection logic stage we propose our new system which exploits the visual appearance of aircraft emergence and common false alarms. }
\label{fig:stagesPropSys}
\end{center}
\end{figure*}

The most promising approaches for vision-based aircraft detection that have been presented in the literature exploit the use of a multi-stage detection pipeline \cite{Carnie2006, Lai2013,Nussberger2014,molloy2017below, jamesCST, JamesRAL}. The first stage typically exploits image pre-processing where spatial and visual aircraft features are highlighted and background clutter is suppressed.  Popular approaches that have been proposed in the literature to achieve this include morphology \cite{Carnie2006,Lai2013,Nussberger2014},   and image frame differencing  \cite{molloy2017below,Nussberger2014}. Various machine learning  \cite{Dey2011} and deep learning \cite{Rozantsev2014,Hwang2018} approaches have  been investigated to exploit the visual appearance of aircraft.  Recently \cite{JamesRAL} used a deep CNN fused with morphological processing in the image pre-processing stage to detect aircraft above the horizon with a mean detection range of  $2527$m and no false alarms. However there are several drawbacks in using learnt approaches in this application as data is expensive and challenging to collect.

In the second stage  temporal filtering approaches are often utilised to emphasise and extract features that possess aircraft-like dynamics such as: Viterbi-based filtering \cite{Barniv, Carnie2006, Lai2013, molloy2017below}, Kalman filtering \cite{Nussberger2014}, and hidden Markov model (HMM) filtering \cite{Barniv,Lai2013, molloy2017below, jamesCST}.

Finally the detection logic stage aims to utilise the information available from the image pre-processing and temporal filtering stages in order to declare whether an aircraft is present or not. Some more simple approaches include looking for a detection in a $5 \times 5$ neighbourhood in the previous frame \cite{Petridis2008} and checking  if the aircraft is increasing in area \cite{Cho2013}.  In \cite{jamesCST} they proposed a new way of modelling aircraft emergence and were able to obtain a theoretically optimal detection logic.   In this paper we aim to build off the theoretically optimal detection logic presented in \cite{jamesCST} however we also aim to exploit the visual appearance of aircraft emergence and common false alarms.

The key contributions of this paper are casting the vision-based aircraft detection problem as an optimal stopping problem where we are able to introduce a penalty  on detecting non-aircraft artefacts such as cloud features.  We are able to establish that the optimal solution occurs on first entry of the change posterior into a stopping region characterised by the union of convex sets. Using these properties we propose a new detection logic stage that is able to improve detection ranges and false alarm rates.

The rest of this paper is structured as follows. In Section \ref{sec:prop} we set up our model of aircraft dynamics and observations, this lets us cast our problem as an optimal stopping problem and  propose two candidate rules for aircraft detection. In Section  \ref{sec:res} we evaluate the performance of our proposed rules on flight data of near collision course encounters. We  provide concluding remarks in \ref{sec:conc}.

%%%%%%%%%%%%%%%%%%%%%%%%%%%%%%%%%%%%%%%%%%%%%%%%%%%%%%%%%%%%%%%%%%%%%%%%%%%%%%%%%%%%%%%%%%%%%%%%%%%%%%

\section{Proposed System}\label{sec:prop}
In this section we describe our proposed vision-based aircraft detection system. Similar to several state of the art approaches we exploit a multi-stage detection pipeline as seen in Figure \ref{fig:stagesPropSys}. An image sensor captures an image which is then stabilised (the data we will test on was stabilised by a GPS-INS sensor see \cite{Bratanov2017} for full details). We will now describe the other stages of our proposed detection system.

\subsection{Image Pre-Processing}
As noted in the introduction morphological processing is commonly exploited in the image pre-processing stage. For a greyscale image $I$, the dilation by a morphological structuring element $S$ is denoted by $I \oplus S$ and erosion is denoted by $I \ominus S$.  Similar to several state of the art approaches \cite{jamesCST, Lai2013} we use bottom hat  $ [I \oplus S] \ominus S -I$  morphological processing which emphasises dark targets
  In \cite{Geyer2009} they observed that targets are generally darker than the background (supporting the notion that bottom hat filters are appropriate).

\subsection{Proposed Temporal filtering}
For our temporal filtering stage we use a HMM approach. 
Let us consider an aircraft in an image sequence which we are trying to detect as soon as possible. The aircraft emerges over time and is (potentially) first visually apparent in an image frame from a single pixel in size.

For $k \geq 0$,  we introduce a Markov chain with a state to represent each of an aircraft's possible $N$ pixel locations in an image.  
Following \cite{jamesCST} we also introduce an extra state to denote when the aircraft is not visually apparent anywhere in the image frame (that is, it has not visually emerged yet, or there is no collision threat). Let us denote this Markov chain as $X_k \in \{e_1, e_2, \dots, e_N, e_{N+1}\}$ where $e_i \in \mathbb{R}^{N+1}$ are indicator vectors with $1$ as the $i$th element and $0$ elsewhere. For $ i \in \{ 1,\dots, N\}$, $e_i$ corresponds to the aircraft being visually apparent at the $i$th pixel and $e_{N+1}$ corresponds to the aircraft not being visually apparent. We denote this the out of image state.

The following are the three key parameters of a HMM description of target image motion for $1 \leq i,j \leq N+1$

\begin{enumerate}
    \item State transition probabilities:\\ $A^{i,j} =P\left( \left. X_{k+1} = e_i \right| X_k = e_j \right)$  gives the probability of moving between different states. This can either be from the out of image state to a pixel in the image, or moving between different pixels.  An aircraft not current located in the image is able to transition from the out of image state to any pixel in the image allowing for the possibility that an  aircraft can visually emerge anywhere as it approaches from a distance. 
    Possible aircraft inter-frame motion can be modelled by a transition patch (see \cite{Lai2013} for detailed explanation of patches). In this paper we will use the patch that allows transitions to side and above pixels as seen in Figure \ref{fig:upTransitions}.
    
    \item Initial probabilities: \\ $\pi^i = P(X_0 = e_i)$ denotes the probability that the target is initially located in state $e_i$. We initialise our filter with an equal probability of $\pi^i  = \frac{1}{N+1}$ for all states.
    
    \item Measurement probabilities: \\ $b^i(y_k) = P(y_k|X_k = e_i)$ specify the probability of obtaining the observed image measurement $y_k$ given that the target is actually in pixel location $e_i$.
    At each time $k>0$ we obtain a noise corrupted, greyscale image $y_k$ that has been morphologically processed using bottom hat. We denote the measurement of the $i$th pixel at time $k$ as $y^i_k$ and following  \cite{Lai2013, jamesCST} for $i,j \in \{1,\dots, N+1 \}$ our diagonal matrix of (unnormalised) output densities is then given by the approximation
     \begin{equation*}
    B^{ij}(y_k) = 
     \begin{cases}
     b^i(y_k) & \text{for } i=j\\
     0 & \text{for } i \neq j\\
     \end{cases}
     \end{equation*}
    where  $b^i(y_k) = y^i_k+1$ for $i\in \{1, \dots, N\}$ and $b^{N+1}(y_k) = 1$ (see \cite{jamesCST} for justification). Essentially, the likelihood of an aircraft is proportional to the strength of the morphology output.
\end{enumerate}

 \begin{figure}
\begin{center}
\includegraphics[scale=0.6,trim={0.0cm 0cm 0.0cm 0.0cm}, clip]{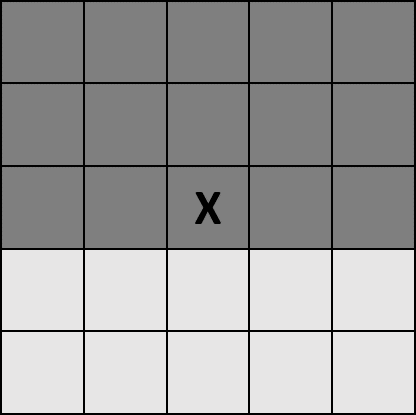}
  \caption{The possible transition of an aircraft at pixel (state) $X$. An aircraft is able to transition to its current pixel (state) or any of the neighbouring pixels beside or above indicated by the darker grey regions. }
\label{fig:upTransitions}
\end{center}
\end{figure}

For $k>0$ we can now calculate   $\hat{X}_k \in R^{N+1}$   via the HMM filter   \cite{elliott1995}, which can be regarded as an indicator of likely target locations

\begin{equation}\label{eqn:CME}
\hat{X}_k= N_{k} B(y_k) A \hat{X}_{k-1}.
\end{equation}
With the initial condition $\pi$ and where $N_{k}$ are scalar normalisation factors defined by 
\begin{equation}\label{eqn:Nk}
N_{k}^{-1} \triangleq \langle \underline{1},B(y_k) A \hat{X}_{k-1} \rangle.
\end{equation}
Here, $\langle \cdot \rangle$ denotes an inner product.

\subsection{Proposed Detection Logic}
We now present our proposed detection logic. We first cast our problem as an optimal stopping rule and establish an optimal policy. We then propose a general pragmatic greedy decision rule for use as our detection logic.

\subsubsection{Problem Formulation}
Our desired detection logic can be characterised by the following continuing and stopping costs, 
\begin{eqnarray}
&&\mathcal{C}(\hat{X}) \triangleq \bar{c}_1 \hat{X}
\\
&& \mathcal{S}_i(\hat{X}) \triangleq  c_2 \hat{X}^{N+1} + \bar{c}_i \hat{X}.
\end{eqnarray}
where $\bar{c}_1 \in  \mathbb{R}^{N+1}$ ($\bar{c}_{1}^{N+1}=0$, $\bar{c}_1^i>0$ for $i \in \{1,\ldots, N\}$) denotes the
delay penalty (when an aircraft is present but a detection has not been declared), $c_2>0$ denotes the false alarm penalty (where an aircraft is not present but a detection has been declared) and $\bar{c}_i \in \mathbb{R}^{N+1}$ for $i\in \{1, \ldots,N\}$ denotes a penalty for incorrectly declaring a detection on a non-aircraft artefact. %misidentification penalty with identification $i\in \{2, \ldots,N\}$ (the misidentification penalty is the penalty for incorrectly declaring where a change has occurred, see Remark \ref{rem:misi} for justification). We  assume that our misidentification penalty is designed to satisfy $\bar{c}_i \hat{X} \ge 0$ and to be small when the correct identification is made. 
% It becomes obvious that the optimal identification decision is given by
%   \begin{equation} \label{equ:isolate}
%       \delta^*\triangleq \text{argmin}_i( \mathcal{S}_i( \hat{X}_\tau) ).
%       %\hat{X}^i_\tau)
%   \end{equation}  
Given the location decision $i$ such that  $\bar{\mathcal{S}}(\hat{X}) \triangleq \min_{i \in \{ 1, \ldots,N\}  } \mathcal{S}_i(\hat{X}_\tau)$, we  seek to design a stopping time $\tau \geq 0$  that minimises the following cost criterion 
\begin{equation} \label{eqn:cost}
     {J} (\tau, \hat{X}) \triangleq 
     \mathbb{E} \left[ \left. \sum_{k=0}^{\tau -1}  
     %\bar{c}_1 \hat{X}_\ell +  c_2 \hat{X}^1_\tau 
     % + c_3 (1- \hat{X}^{\delta_{\tau}}_\tau) 
     \mathcal{C}(\hat{X}_k) 
+ \bar{\mathcal{S}}(\hat{X}_\tau) 
     \right|   \hat{X}  \right].
\end{equation}
Where  $\mathbb{E} \big[ \cdot \big| \hat{X} \big]$ denotes the expectation operation corresponding to the probability measure where the initial state has distribution $\hat{X}$,

% \begin{remark}\label{rem:misi}
% The misidentification term  allows penalties on undesirable attributes beyond those typically considered. An interesting example is a posterior's variance penalty investigated in \cite[Ch. 13]{krishnamurthy2016}. 
% In following sections we will see the performance benefits of 
% moving beyond simple probability misidentification  penalties in our vision-based aircraft detection application
% \end{remark}

\subsubsection{Optimal Policy}
Let us consider a stopping action $u_k \in \{1  \text{ (continue)}, 2  \text{ (stop)}\}$. Then, there is an optimal policy $\mu^*(\hat{X}_k)$ to minimise our cost criterion \eqref{eqn:cost} given by the 
value function $V(\hat{X}_k) \triangleq \min_{\tau} \{ {J}(\tau, \hat{X}_k)\}$. This value function can be described by the following  recursion (similar to \cite[pg. 258]{krishnamurthy2016})
\begin{equation}\label{eqn:val}
    \begin{split}
        V(\hat{X}_k)        = \min& \left\{\mathcal{C}(\hat{X}_k) \right.  \\
        &+ \left. \mathbb{E} \left[ \left. V\left(\hat{X}^+(\hat{X}_k,y)\right)  \right| \hat{X}_k  \right] , \right.
        \left. \bar{\mathcal{S}}(\hat{X}_k)  \right\}
    \end{split}
\end{equation}
where $\hat{X}^+(\hat{X},y) = \langle \underline{1},B(y) A \hat{X} \rangle^{-1} B(y)A \hat{X} $, and $B(y) = \text{diag}(b^1(y),\dots ,b^{N+1}(y))$. 
If we let $Q(\hat{X}) \triangleq \mathcal{C}(\hat{X}) + \mathbb{E}[ V(\hat{X}^+ (\hat{X} , y)) |\hat{X}]$ denote the total cost incurred if continuing, then the optimal policy is given by
\begin{equation} \label{equ:policy}
\mu^*(\hat{X}_k) = \begin{cases} 
      1 \text{ (continue)} & \text{if}  \quad Q(\hat{X}_k) < \bar{\mathcal{S}}(\hat{X}_k)  \\
      2 \text{ (stop)} &  \text{if} \quad  Q(\hat{X}_k) \ge \bar{\mathcal{S}}(\hat{X}_k).  
    \end{cases}
\end{equation}
% with identification rule \eqref{equ:isolate}.

In general, the value recursion \eqref{eqn:val} is difficult to compute making it unsuitable for use in the vision-based aircraft detection application. We can however establish some properties of the optimal solution that allow us to propose some practical rules. 

\subsubsection{Structure of the Optimal Policy}
We now establish that our optimal stopping problem can be solved by finding $N$ convex stopping sets  (rather than solving the dynamic programming recursion equations directly).  We define the stopping region for state $e_i$ as follows  %$\mathcal{S}_i(\hat{X}) \triangleq c_2 \hat{X}^1 + c_3 (1- \hat{X}^{i})$ and
$\mathcal{R}^i_S\triangleq \{\hat{X}: S_i(\hat{X}) \le Q(\hat{X})  \}$.

\begin{theorem}\label{thm:sets}
Consider the value recursion \eqref{eqn:val} and let  $\mathcal{R}_S \triangleq \cup_{i\in \{ 1, \ldots,N\}} \mathcal{R}^i_S $ be the union of the $N$ regions $\mathcal{R}_S^i$. Then, the optimal stopping time given by 
\[
\tau^*\triangleq \inf \{ k : \hat{X}_k \in \mathcal{R}_S \}
\]
Furthermore, the regions $\mathcal{R}_S^i$ forming the optimal stopping region $\mathcal{R}_S$ are convex and contain $e_i$. 
\end{theorem}
\begin{proof}
See Appendix for proof. 
\end{proof}

Theorem \ref{thm:sets}  establishes that the solution to our optimal stopping problem can be solved by finding the $N$ convex stopping sets $\mathcal{R}^i_S$. Given the parallels with our problem setup and quickest change detection and identification,  we expect the stopping region $\mathcal{R}_S$ to not be a connected region in general \cite{Dayanik}.

\subsubsection{A Greedy Decision Rule}
We now propose a practical greedy rule which can be applied in our application and establish some performance bounds.

Let us define $\mathcal{R}_g^i\triangleq \{\hat{X}: S_i(\hat{X}) \le C(\hat{X})  \}$ and the union of sets
$\mathcal{R}_g =\cup_{i\in \{ 1, \ldots,N\}} \mathcal{R}_g^i$.
We also define the probability of a false alarm (PFA)  as  $\text{PFA} \triangleq P(X_\tau=e_{N+1})$.

\begin{lemma}\label{lem:greedy}
For $i\in \{ 1, \ldots,N\}$, we consider  the union of sets $\mathcal{R}_g^i \subset \mathcal{R}^i_S$ and the value recursion
\eqref{eqn:val}. Then, the greedy decision rule given by
\[
\tau^g \triangleq \inf \{ k : \hat{X}_k \in  \mathcal{R}_g \}
\]
 achieves the performance bound 
\[
\text{PFA}  \le \frac{c_m}{c_m+c_2},
\]
where the maximum delay constant $c_m \triangleq \max_i \bar{c}_1^i$.
% and the identification decision $\delta^*$ is given by \eqref{equ:isolate}.

\end{lemma}

\begin{proof} 
See Appendix for proof

\end{proof}
 
Lemma \ref{lem:greedy} suggests that a pragmatic solution to our optimal stopping problem  is the greedy (sub-optimal) stopping rule $\tau^g$.  Moreover, if this rule is used then the cost function parameters can be related to a false alarm performance trade-off.   The lemma also provides some insight into the role of cost parameters and false alarm performance.  

Now, rather than solving the  dynamic programming recursion equations \eqref{eqn:val} directly or trying to calculate the convex stopping sets $\mathcal{R}^i_S$ (both of which are computationally expensive), we are able to just run the HMM filter \eqref{eqn:CME}. We highlight that you do not need to know the value of the costs constants $\bar{c}_1, c_2, \bar{c}_i$ to run the filter. 

\subsubsection{Proposed Greedy Detection Logic}
We now present two proposed greedy rules for use in the detection logic stage in our multi stage detection pipeline. Due to Lemma \ref{lem:greedy} we know that these stopping rules are a greedy (sub-optimal) solution to our cost criterion \eqref{eqn:cost}.\\
\textbf{Greedy Rule 1:}
We first consider a a simple stopping rule of the form
\begin{equation}
    \tau^{g_1} = \text{inf} \left\{k>0 : \max_{i\in \{1, \dots, N\}}{\hat{X}^i_k} \geq h^{g_1} \right\}, 
\end{equation}
where $\hat{X}^i_k$ is the $i$th element of $\hat{X}_k$.
Intuitively, this rule declares a detection when the probability of being in one of the first $N$ states is higher than a threshold. 

%  with $\bar{c}_1^i=c_1$ for all $i\in\{1,\ldots,N\}$,  and
% $\bar{c}_i=c_3(\underline{1}-e_i)'-(c_1+c_2)e_1'$ with $c_3=c_1/(1-h^{g_1})$ for some $c_1,c_2>0$ such that $c_3> (c_1+c_2)$ .

\textbf{Greedy Rule 2:}
We  introduce a mapping $ \mathcal{M}(\cdot)$ which reshapes the pixel elements of the vector $\hat{X}_k$ to an image matrix %.
%where $i,j$ correspond to the image pixel row and column location. 
%We 
and also introduce the inverse mapping $\mathcal{M}^{-1}(\cdot)$.
Our second stopping rule is of the form
\begin{equation}
       \tau^{g_2} = \text{inf} \left\{k>0 : \max_{i\in \{1, \dots, N\}} {{\zeta}^i(\hat{X}_k)} \geq h^{g_2} \right\},
\end{equation}
where $\zeta^i(\cdot)$ is the $i$th element of $\zeta(\cdot)$ and ${\zeta} (\cdot)\triangleq \mathcal{M}^{-1}\left( \mathcal{M}( \cdot) * \omega\right)  $ and $*$ denotes the convolution operation with the kernel
\begin{equation*}
    \omega = \begin{bmatrix}-1& -1 &-1\\ -1 &1& -1 \\ -1 &-1 &-1  \end{bmatrix}.
\end{equation*}

Intuitively we know that aircraft emerge in an image at approximately 1 pixel in size and hence only occupy one state. As our aim is to detect this aircraft emergence, we propose a kernel that penalizes the states around the aircraft as a type of  penalty. Additionally, we found that this kernel was an effective way to reject cloud artifacts that are generally larger in size (more pixels) than an aircaft emerging in an image sequence. 

% Due to Lemma \ref{lem:greedy} we know this stopping rule is also a greedy solution to  our stopping problem 
% with $\bar{c}_1^i=c_1$ for all $i\in\{1,\ldots,N\}$,  and
% $\bar{c}_i=c_3(\underline{1}-{\zeta}(e_i))'-(c_1+c_2)e_1'$ with $c_3=c_1/(1-h^{g_2})$ for some $c_1,c_2>0$ such that $c_3> (c_1+c_2)$ (noting that $\zeta^i{(\hat{X}) = \zeta^i(e_i)'\hat{X}}$).

%%%%%%%%%%%%%%%%%%%%%%%%%%%%%%%%%%%%%%%%%%%%%%%%%%%%%%%%%%%%%%%%%%%%%%%%%%%%%%%%%%%%%%%%%%%%%%%%%%%%%%
\section{Results}\label{sec:res}
In this section we examine the performance of our two proposed detection logic stages in the vision-based aircraft detection application. As a baseline we compare to the theoretically optimal intermittent signal detection (ISD) rule presented in \cite{jamesCST}, we denote this the baseline ISD rule.

\begin{figure}
\begin{center}
\includegraphics[scale=0.5, trim={0.0cm 0.0cm 0.0cm 0.0cm}]{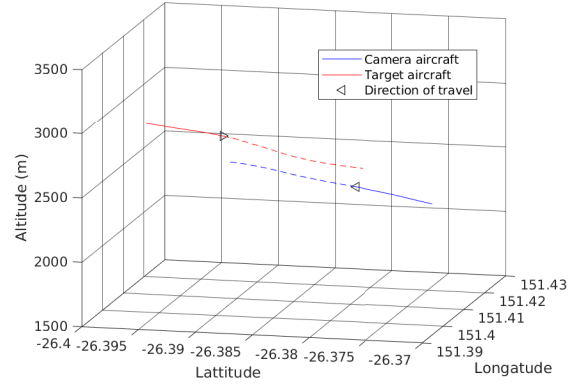}
\caption{An example flight path of a camera and target aircraft in a head on near collision course encounter.}  
\label{fig:illustrative}
\end{center}
\end{figure}

We compare the performance of our proposed detection systems on 15 near mid-air collision course encounters between two fixed wing aircraft; the data collection aircraft was a ScanEagle UAV  and the other aircraft was a Cessna 172. An example of the flight path of a head on near collision course encounter is presented in Figure \ref{fig:flightPath}, and Figure \ref{fig:exampleAircraft} is an example of the size of aircraft we are detecting. For full details of the flight experiments see \cite{Bratanov2017}.

\subsection{Illustrative Example}
\begin{figure}[t!]
\begin{center}
\includegraphics[scale=0.6, trim={0.0cm 1.5cm 0.0cm 0.0cm}]{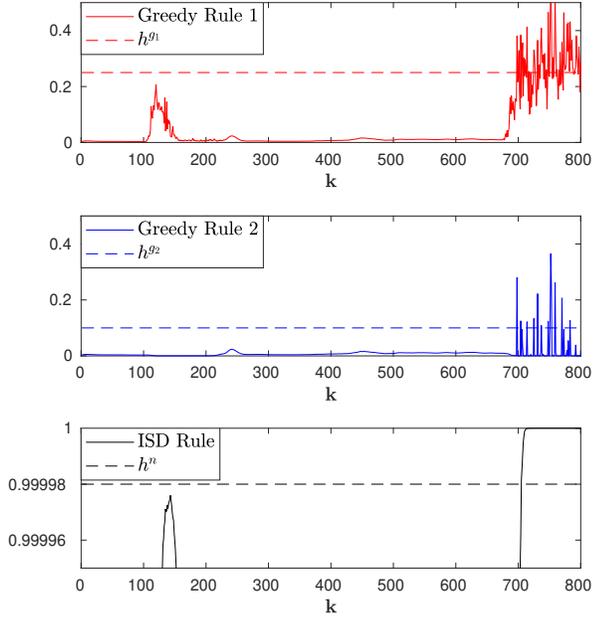}
\caption{An illustrative example of the proposed Greedy rules compared to the baseline ISD rule for Case 10. The aircraft is visually present from image frame $k = 680$ to $ k = 800$. The peak in the Greedy rule 1 test and ISD rule test statistics just after $k=100$ corresponds to cloud artefacts.}  
\label{fig:flightPath}
\end{center}
\end{figure}

For $i \in \{1,\dots,N\}$ Figure \ref{fig:illustrative} shows an illustrative example for Case 10 of the Greedy rule 1 test statistic $(\max_i{\hat{X}^i_k})$ the Greedy rule 2 test statistic $(\max_i{{\zeta}^i})$ and the baseline ISD rule test statistic $(1-\hat{X}^{N+1}_k)$.
The aircraft is visually present from image frame $k = 680$ to $ k = 800$. All 3 test statistics effectively increase at a similar time corresponding to them detecting aircraft presence. Greedy rule 1 and  Greedy rule 2 both correctly declare at $k=698$ (corresponding to aircraft at range $2542.4$m) and the baseline ISD rule correctly declares  slightly later at $k=705$ (corresponding to aircraft at range $2479.8$m).

We highlight the peak in the Greedy rule 1  and ISD rule test statistics just after $k=100$. This corresponds to the rules mistaking small cloud features for aircraft. Importantly, Greedy rule 2 effectively rejects this false alarm.

We also note the choice of thresholds for the 3 rules. Greedy rule 1 has a higher threshold than Greedy rule 2 so that it does not declare false alarms. We also note that the scale of the baseline ISD rule corresponds to a significantly higher  threshold, and to a higher number of significant figures than the  two greedy rules, which can be set between 0 and 1.

\subsection{Detection Range Performance}
\begin{figure}
\begin{center}
\includegraphics[scale=0.5,trim={0.0cm 0.0cm 0.0cm 0.0cm}, clip]{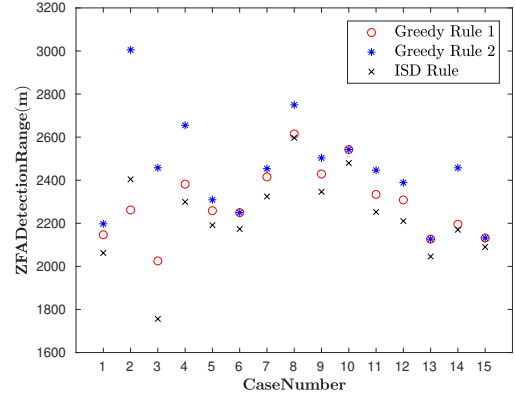}
  \caption{A comparison of the proposed Greedy rules compared to the baseline ISD rule for all 15 cases presented in \cite{Bratanov2017}. The mean detection distance and standard error  was $2295$m and $22$m for Greedy rule 1, $2445$m and  $61$m for Greedy rule 2 and  $2227$m and $52$m for the baseline ISD rule.}
\label{fig:comparison}
\end{center}
\end{figure}

We now investigate the detection range performance across all 15 cases from our test data. 
We highlight that the detection range and false alarm rates vary with the choice of the threshold parameters.
Hence, to ensure fair comparison, we will compare these rules on the basis of the lowest thresholds for each rule that achieve zero false alarms (ZFAs) in this data set.  In practice, detection thresholds could be adaptively selected on the basis of scene difficulty such as proposed in  \cite{molloy2017adapt}.

The resulting ZFA detection ranges are presented in Figure \ref{fig:comparison}.  The mean detection distance and standard error was $2295$m and $22$m for Greedy rule 1, $2445$m and  $61$m for Greedy rule 2 and  $2227$m and $52$m for the baseline ISD rule. Both our Greedy rules improved performance relative to the baseline ISD rule but importantly Greedy rule 2  improved detection ranges  by a mean distance of $218$m (9.8\%).

Intuitively, the performance benefits from Greedy rule 2 are from penalising the non-aircraft artefacts and as a result more effectively rejecting false alarms. This allows for a lower ZFA threshold to be set and therefore an earlier detection can be declared.

%%%%%%%%%%%%%%%%%%%%%%%%%%%%%%%%%%%%%%%%%%%%%%%%%%%%%%%%%%%%%%%%%%%%%%%%%%%%%%%%%%%%%%%%%%%%%%%%%%%%%%
\section{Conclusion}\label{sec:conc}

In this paper we presented a new technique for the modelling and detection of aircraft above the horizon that is able to penalise non-aircraft artefacts (such as textured clouds and other unstructured terrain). We first posed the vision-based aircraft detection problem in a Bayesian setting and established some optimal properties.  We then proposed a practical greedy rule and developed some bounds for characterising its performance. 
Finally, we investigated the performance of two greedy rules on real flight data where we were able to improve detection ranges  by a mean distance of $218$m (9.8\%) relative to a current state of the art vision-based aircraft detection technique.

%%%%%%%%%%%%%%%%%%%%%%%%%%%%%%%%%%%%%%%%%%%%%%%%%%%%%%%%%%%%%%%%%%%%%%%%%%%%%%%%
%\section{ACKNOWLEDGEMENTS}

%\addtolength{\textheight}{-2.5cm}   % This command serves to balance the column lengths

\bibliographystyle{IEEEtran}
\bibliography{IEEEabrv,ref}

%%%%%%%%%%%%%%%%%%%%%%%%%%%%%%%%%%%%%%%%%%%%%%%%%%%%%%%%%%%%%%%%%%%%%%%%%%%%%%%%

\appendix

\subsection{Proof of Theorem 1}
For $i\in \{ 1, \ldots,N\}$, we define the cost for the location  decision $i$ as
\begin{equation}
     {J}_i (\tau, \hat{X}) \triangleq       \mathbb{E} \left[ \left. \sum_{k=0}^{\tau -1}  
     \mathcal{C}(\hat{X}_k)+ \mathcal{S}_i(\hat{X}_\tau)
     \right|   \hat{X}  \right]
\end{equation}
and consider the value function $V_i(\hat{X}_k) \triangleq \min_{\tau} \{ {J}_i(\tau, \hat{X}_k)\}$ for the stopping time
recursion described by 
 \begin{equation}%\label{eqn:val}
    \begin{split}
        V_i(\hat{X}_k) =       \min& \left\{\mathcal{C}(\hat{X}_k) \right. \\
         & \left. + \mathbb{E} \left[ \left. V_i\left(\hat{X}^+(\hat{X}_k,y)\right)  \right| \hat{X}_k  \right] , \right.
        \left. \mathcal{S}_i(\hat{X}_k)  \right\}.
    \end{split}
\end{equation}
Noting that the cost is linear here, then according to  \cite[Theorem 7.4.2]{krishnamurthy2016}, $V_i(\hat{X}_k)$ are 
concave in $\hat{X}$. %on $\Pi$ 
Moreover, our location decision $\bar{\mathcal{S}}(\hat{X}) \triangleq \min_{i \in \{ 1, \ldots,N\}  } \mathcal{S}_i(\hat{X}_\tau)$ gives that
 at each $\hat{X}_k$, ${J} (\tau, \hat{X})=\min_i {J}_i (\tau, \hat{X})$, and hence $V(\hat{X}_k) = \min_i V_i(\hat{X}_k)$ (swapping the order of the minimisation operations) implies that $V(\hat{X}_k)$ is 
concave in $\hat{X}$ (concavity is preserved under minimum operations).
The $V(\hat{X}_k)$ concavity gives that $\mathcal{R}^i_S$ are convex sets (see similar proof steps in  \cite[Theorem 12.2.1]{krishnamurthy2016}).

 We proceed by considering the $N$ location decisions as individual optimal stopping problems. For $i\in \{ 1, \ldots,N\}$ if $\mathcal{S}_i(e_i)=0$ then $e_i \in \mathcal{R}^i_S$. 
When $\hat{X} \in \mathcal{R}^i_S$ for any $i$ implies 
$\mathcal{S}_i (\hat{X}) \le Q (\hat{X})$ and under optimal rule \eqref{equ:policy} implies that the optimal action is to stop. Contrarily, $\hat{X} \notin \mathcal{R}^i_S$ for any $i$ implies 
$\mathcal{S}_i (\hat{X}) > Q (\hat{X})$ and under optimal rule \eqref{equ:policy} implies the optimal action is to continue. Hence the union of the $\mathcal{R}^i_S$ regions defines the optimal stop region as given in the theorem statement.

\subsection{Proof of Lemma 1}
For $i\in \{ 1, \ldots,N\}$,  $ \mathcal{C}(\hat{X}_k) \geq \mathcal{S}_i(\hat{X}_k)$ implies that  $Q(\hat{X}_k) \ge \mathcal{S}_i(\hat{X}_k)$  and therefore, according to the optimal policy \eqref{equ:policy}, we should stop when $\hat{X}_k \in \mathcal{R}_g^i$.
Further,  if $\mathcal{S}_i(e_i)=0$ 
%gives $V(e_i)=0$ 
then $e_i \in \mathcal{R}_g^i$ giving our first lemma result.  
%Noting $\mathcal{S}(e_1)>0$ and $\mathcal{C}(e_1)=0$ gives we should continue at e_1.
%gives that $V(\hat{X}_k)$ have stopping regions at the $N-1$ corresponding vertices (and 

 Simple algebra shows that $ \mathcal{C}(\hat{X}_k) \geq \mathcal{S}_i(\hat{X}_k)$ implies
%$(c_1+c_2) \hat{X}^1_k - c_3 \hat{X}^\delta_k < (c_1-c_3)$.
$(c_m+c_2) \hat{X}^{N+1}_k + \bar{c}_i \hat{X}_k
%c_3 (1-\hat{X}^\delta_k) 
< c_m$ (given that $c_m(1-\hat{X}^{N+1}_k) \geq \bar{c}_1 \hat{X}_k$).  Taking the expectation operation on both sides and  using the idempotent property gives our lemma's PFA performance bound result.

% When $\bar{c}_i=c_3(\underline{1}-e_i)'$, we instead write
% $(c_m+c_2) \hat{X}^1_k + %c_\delta \hat{X}_k
% c_3 (1-\hat{X}^\delta_k) 
% < c_m$ and similar steps give the two additional performance bound results.

%\addtolength{\textheight}{-3cm}   % This command serves to balance the column lengths

\end{document}